\documentclass[letterpaper]{article}
\usepackage[margin=1.0in]{geometry}
\usepackage{amsmath,amsthm,amssymb,amsfonts,mathtools}
\usepackage[linesnumbered]{algorithm2e}

\usepackage{graphicx}
\usepackage{textcomp}
\usepackage{xcolor}
\usepackage{cite}
\usepackage[hyphens]{url}
\usepackage{enumitem}
\usepackage{pdfsync}
\newcommand{\mc}{\mathcal}
\newcommand{\mb}{\mathbb}
\newcommand{\X}{\mc{X}}
\newcommand{\la}{\langle}
\newcommand{\ra}{\rangle}

\DeclareMathOperator*{\amin}{argmin}

\DeclareMathOperator*{\E}{\mathbb{E}}
\DeclareMathOperator{\R}{\mathbb{R}}
\DeclareMathOperator*{\proj}{proj}

\newcommand{\query}{\delta}
\newcommand{\num}{n}
\usepackage{xcolor}
\usepackage[sort,compress]{natbib}
\SetKwInput{KwInput}{Input}                %
\SetKwInput{KwOutput}{Output} 
\newtheorem{assumption}{Assumption}
\newtheorem{theorem}{Theorem}
\newtheorem{lemma}{Lemma}

\newtheorem{corollary}{Corollary}
\title{Improved Rates for Derivative Free Gradient Play\\ in Strongly Monotone
    Games\thanks{Drusvyatskiy's research was supported by NSF
    DMS-1651851 and CCF-2023166 awards. Ratliff's research was supported by NSF CNS-1844729 and Office of Naval Research YIP Award N000142012571. Fazel's research was supported in part by awards NSF TRIPODS II-DMS 2023166, NSF TRIPODS-CCF 1740551, and NSF CCF 2007036.}
}
\author{Dmitriy Drusvyatskiy\thanks{Drusvyatskiy is in the Department
    of Mathematics
    at the University of Washington, Seattle. email: {\tt ddrusv@uw.edu}.}, Maryam
Fazel\thanks{Ratliff and Fazel are in the Department of Electrical and
Computer Engineering at the University of Washington, Seattle.
email: {\tt$\{$ratliffl,mfazel$\}$@uw.edu}.}, Lillian J.
Ratliff$^\ddag$}

\date{}

\usepackage{abstract}

\begin{document}

\maketitle

\begin{abstract}
The influential work of~\citet{bravo2018bandit} shows that derivative free
gradient play in strongly monotone games has complexity $O(d^2/\varepsilon^3)$, where $\varepsilon$ is the target accuracy on the expected squared distance to the solution. This paper shows that the efficiency estimate is actually
$O(d^2/\varepsilon^2)$, which reduces to the known efficiency guarantee for the
method in unconstrained optimization. The argument we present  simply
interprets  the method as stochastic gradient play on a slightly perturbed strongly monotone game 
to achieve the improved rate.
\end{abstract}

\section{Introduction}

Game theoretic abstractions are foundational in many application domains ranging from machine
learning to reinforcement learning to control theory. For instance, in machine learning, game theoretic abstractions
are used to develop solutions to learning from
adversarial or otherwise strategically generated data (see, e.g.,
\cite{madry2018towards,fiez2020implicit,goodfellow2014generative,narang2022multiplayer}). Analogously,
in reinforcement learning and control theory, game theoretic abstractions are used to
develop robust algorithms and policies. 
 (see, e.g.,
\cite{ratliff2016characterization,zhou2017mirror,ratliff2020adaptive,li2011designing,yekkehkhanyadversarial,Zhang2019MultiAgentRL,zhang2019policy}). 
Additionally, game theoretic abstractions are used to capture interactions between
multiple decision making entities 
and to model asymmetric information and
incentive problems  (see, e.g.,
\cite{ratliff2020adaptive,li2011designing,savas2019incentive}).

In such game theoretic abstractions, each decision-maker or `player' faces an
optimization problem that is dependent on the decisions of other players in the
game. Learning as a tool for finding equilibria, or explaining how players in a game
arrive at an equilibrium through a process of t\^atonnement, is a long studied
phenomenon~\cite{fudenberg1998theory,cesa2006prediction}. 
Gradient-based learning algorithms form a very natural class of
learning algorithms for games on continuous actions spaces with sufficiently
smooth cost functions. 
Additionally, in both control theory and machine learning, 
typically gradient-based methods are used since they scale well.

There is an extensive body of literature---too vast to cite all relevant work---analyzing
stochastic gradient play and its variants in different classes of games, ranging from
zero-sum to general-sum. The majority of the work on stochastic gradient play
assumes access to a gradient oracle that provides an unbiased estimate
of each player's individual gradient---i.e., the partial gradient of a player's cost with
respect to their own choice variable.

Counter to this, we are motivated by settings in which non-cooperative players interact in extremely low-information  environments: in this
paper, we examine the long-run behavior of learning with so-called ``bandit
feedback'' in strongly monotone games. Specifically,  players have access only
to a loss function
oracle, and use responses to queries (of the loss function value) to construct a gradient estimate. The
bandit feedback setting has been studied extensively in the single player
case \cite{agarwal2010optimal,shamir2017optimal,flaxman2004online,nesterov2017random}, and over the last few years it has been extended to the multi-agent
setting \cite{bravo2018bandit,tatarenko2020bandit,tatarenko2019learning}. 

In particular, the influential work of \citet{bravo2018bandit} studies derivative free ``gradient play'' wherein players
formulate a gradient estimate using a  single-point query to their loss
function. While in a general game such algorithms (even with perfect gradient information) may not converge, 
for strongly monotone games, which admit a unique Nash equilibrium, the authors show convergence to the Nash equilibrium. Moreover, they show that the iteration complexity is  $O(d^2/\varepsilon^3)$, where $\varepsilon$ is the target accuracy on the expected squared distance to the solution, and $d$ is the problem dimension. 
It was conjectured in \cite{bravo2018bandit} that this rate should match that of single player optimization,
which is known to be $O(1/\varepsilon^{2})$ in terms of target accuracy.

In this paper, we resolve this open question by showing that the iteration complexity
is in fact $O(d^2/\varepsilon^2)$. Our proof deviates significantly from the analysis
in \citet{bravo2018bandit}. In particular, we take the unique perspective that
the update players are executing is simply stochastic gradient play on a
slightly perturbed strongly monotone game, 
and this tighter analysis leads to the optimal rate result.

\section{Problem Setup and Algorithm}

In this paper, we consider an $\num$-player game defined by cost functions 
 $f_i\colon\X_i\to \mb{R}$ and sets of strategies  $\X_i\subset\R^{d_i}$. Thus
 each player $i\in \{1,\ldots, \num\}$ seeks to solve the problem
\[\min_{x_i\in \X_i}f_i(x_i,x_{-i}),\]
where $x_{-i}$ denotes the actions of all the players excluding player $i$.  A
vector of strategies $x^\star=(x_1^\star,\ldots, x_\num^\star)$ is a {\em Nash equilibrium}  if 
 each player $i$ has no incentive to unilaterally change their strategy, that is 
 \[x_i^\star\in \amin_{x_i\in \X_i}f_i(x_i,x_{-i}^\star).\]
 Throughout, the symbol $\nabla_i(\cdot)$  denotes the partial derivative of the
 argument
 $(\cdot)$ with
respect to $x_i$.  Set $d=\sum_{i=1}^{\num} d_i$, and let $\mathbb{S}_i$ and $\mathbb{B}_i$ denote the unit sphere and unit ball in $\mathbb{R}^{d_i}$, respectively.
Additionally, throughout we impose the following convexity and smoothness assumptions.

\begin{assumption}[Standing]\label{assump:convex}
{\rm
	There exist constants $\alpha,\beta, L\geq 0$ such that for each $i\in
    \{1,\ldots,\num\}$, the following hold:
    \begin{enumerate}[itemsep=0pt,topsep=5pt]
        \item[(a)] \label{it1} The set $\X_i \subset \mb{R}^{d_i}$ is closed and convex, and there exist constants $r,R>0$ 
            satisfying $r\mb{B}\subseteq \X\subseteq R\mb{B}$ where
            $\mc{X}:=\mc{X}_1\times \cdots \times \mc{X}_n$.
        \item[(b)] \label{it2} The function $f_i(x_i,x_{-i})$ is convex and $C^1$-smooth in $x_i$ and the gradient 
		$\nabla_{i} f_i(x)$ is $\beta$-Lipschitz continuous in $x$.
    \item[(c)] \label{it3} The Jacobian of the map $\nabla_i f_i(x)$ is $L$-Lipschitz continuous, meaning
		\[\|\nabla (\nabla_i f_i)(x)-\nabla (\nabla_i f_i)(x')\|_{\rm op}\leq
        L\|x-x'\|,\quad \forall\ x,x'\in\X.\]
    \item[(d)] \label{it4} The map defined by the vector of individual partial
        gradients,
        \[g(x):=(\nabla_{1}f_1(x),\nabla_{2}f_2(x),\ldots,
        \nabla_{\num}f_\num(x)),\] is $\alpha$-strongly monotone on
        $\X$, meaning that the following inequality holds:
		\[\langle g(x)-g(x'),x-x'\rangle\geq \alpha \|x-x'\|^2\qquad \forall
        x,x'\in \X.\]
    \item[(e)]  \label{it5} For each $i\in \{1,\ldots,\num\}$ and $x\in \X$, the function $f_i$ satisfies $|f_i(x)|<\infty$
        and we set $F_\ast:=\max_i \max_{x\in \mc{X}}|f_i(x)|^2$.
	\end{enumerate}}
\end{assumption}

Items (a)--(b) and (d)--(e) are identical to those in
\citet{bravo2018bandit}. In contrast, item (c) is not assumed in
\citet{bravo2018bandit}, but will be important in what follows. We note that
when applied to the single player setting $\num=1$, none of our results require
item (c) and it may be dropped entirely.

Classical results, such as those in the seminal work by \citet{rosen1965existence}, guarantee that the game admits a unique Nash equilibrium under Assumption~\ref{assump:convex}. 
In this work, we study a derivative-free algorithm proposed by \citet{bravo2018bandit}
for finding the Nash equilibrium of the game.  We note that the gradient
estimator used by \citet{bravo2018bandit} is motivated by the analogous
construction introduced by \citet{flaxman2004online} for the single player setting. 
 The procedure is recorded as Algorithm~\ref{algo:bravo}.

 \RestyleAlgo{ruled}
 \begin{algorithm}[t!]
    \caption{Derivative Free Gradient Play~\citep{bravo2018bandit}: \texttt{DFO}$(x_0, \{\eta_t\}_{t\geq 1},\delta, T)$}\label{algo:bravo}
    \KwInput{Horizon $T\in\mathbb{N}$, step-sizes $\eta_t>0$, radius $\delta\in(0,r)$, initial strategies $x^0\in (1-\delta)\X$.}
    \For{$t=0,\ldots,T-1$}{
        \For{$i=1,\ldots, n$}{
            Sample $v_{i}^{t}\in {\mathbb S}_i$ uniformly at random\;
            
            Play $x^{t}_{i}+\query v_{i}^{t}$\;

            Compute $\hat g_i^t=\frac{d_i}{\query}f_i(x_{i}^t+\query
      v_{i}^t,x_{-i}^t+\query v_{-i}^t)v_{i}^t$\;
      
      Update $x^{t+1}_i=\proj_{(1-\delta)\X_i}\left(x_i-\eta_t \hat
      g_i^t\right)$
}
}
\KwOutput{$x^T=(x_1^T,\ldots, x_n^T)$}
\end{algorithm}

In each iteration $t$, Algorithm~\ref{algo:bravo} samples $v^{t}\in
\mathbb{S}_1\times\cdots\times \mathbb{S}_\num$ uniformly at random and then declares 
\begin{equation}
x^{t+1}=\proj_{(1-\query)\X}(x^t-\eta_t \widehat{g}_{t})
    \label{eq:zo_plain_F}
\end{equation}
where  \[\widehat{g}_t:=(\hat g_1^t,\ldots, \hat g_\num^t).\] 
The reason for
projecting onto the set $(1-\delta)\X$ is simply to ensure that in the next
iteration $t+1$, the action profile is valid in the sense that
$x^{t+1}_{i}+\query v_{i}^{t+1}$  lies in $\X_i$ for each player
$i\in\{1,\ldots, n\}$. 

\citet{bravo2018bandit} show that with appropriate parameter choices,
 Algorithm~\ref{algo:bravo} will find a point $x$ satisfying $\E[\|x-x^\star\|^2]\leq
\varepsilon$ after $O(\frac{d^2}{\varepsilon^3})$ iterations, and leave it as an open question if this result is tight. We provide a different convergence argument that yields an improved
 efficiency estimate $O(\frac{d^2}{\varepsilon^2})$. The estimate matches the known rate of
 convergence of the method for unconstrained optimization problems (i.e., $\num=1$ and $\X=\mathbb{R}^d$) established by
 \citet{agarwal2010optimal}.\footnote{Formally, the paper of
     \citet{agarwal2010optimal} shows  the bound of
     $O\left(\frac{d^2}{\varepsilon^2}\right)$ on the squared distance of the average iterate to the minimizer because the analysis is based on bounding the regret.} We note that compared to  \cite{bravo2018bandit} our results do rely on a slightly stronger assumption on the second-order smoothness of the loss functions, summarized in item \ref{it3} of Assumption~\ref{assump:convex}. This assumption is not needed in the single player setting $\num=1$.
 \begin{theorem}[Informal]\label{thm:informal_main}
 For sufficiently small $\varepsilon>0$, there exists a choice of $\delta>0$ and
 $\eta_t>0$ such that Algorithm~\ref{algo:bravo}  generates a sequence $x^t$
 satisfying \[\E\|x^t-x^\star\|^2\leq \varepsilon\] for $t$ on the order of $\frac{d^2}{\varepsilon^2}$.
 \end{theorem}

 We provide the formal statement of the main result as well as the proof in the next section. A key idea is to interpret Algorithm~\ref{algo:bravo} as stochastic gradient play applied to a slightly perturbed strongly monotone game. 
 
 \section{Main Result}
 
 The starting point is the very motivation for the update
 \eqref{eq:zo_plain_F}, which is that the vector $\widehat{g}_t$ is an unbiased estimator of the gradient map for a \emph{different} game. Namely, for each player $i$, define the smoothed cost
\[
f_i^\query(x_i,x_{-i}):=\E_{w\sim U_i}f_{i}(x_i+ \delta w_i,x_{-i} +\delta w_{-i}),
\]  
where $U_i$ denotes the uniform distribution on \[\mathbb{B}_i\times
\left(\bigtimes_{j\neq i} {\mathbb S}_j\right).\] Thus $w$ is a vector of size
$d$, where recall that $d=\sum_{i=1}^n d_i$. The following is proved in  \cite[Lemma C.1]{bravo2018bandit} and follows closely the argument in \cite{flaxman2004online}.

\begin{lemma}[Unbiased gradient estimator]
    For each index $i\in\{1,\ldots, n\}$, let  $v_{i}$ be sampled uniformly from
    $\mathbb{S}_i$ and define the random vector \[\hat
        g_i=\frac{d_i}{\query}f_i(x_{i}+\query v_{i},x_{-i}+\query
    v_{-i})v_{i}.\] The following equality holds,
    \[\E[\hat g_i]=\nabla_i f^{\delta}_i(x).\]
\end{lemma}

The path forward is now clear: we interpret Algorithm~\ref{algo:bravo} as stochastic gradient play on the perturbed game defined by the 
losses $f_i^{\delta}$ over the smaller set $(1-\delta)\X$.  
To this end, define the perturbed gradient map \[g^{\delta}(x)=(\nabla_1
f_1^{\delta}(x),\ldots,\nabla_\num f_\num^{\delta}(x)).\] Lemmas~\ref{lem:smooth_game} and \ref{lem:strong_mon} estimate the smoothness and strong monotonicity constants of the perturbed game, thereby allowing us to invoke classical convergence guarantees for stochastic gradient play on the perturbed game. 

\begin{lemma}[Smoothness of the perturbed game]\label{lem:smooth_game}
    For each index $i\in\{1,\ldots, n\}$, the loss $f_i^{\delta}(x)$ is differentiable and the map $x\mapsto \nabla_i f^{\delta}_i(x)$ is $\beta$--Lipschitz continuous. Moreover the following estimate holds,
    \[\|g(x)-g^{\delta}(x)\|\leq \beta \delta \num\qquad \forall x\in \X.\]
\end{lemma}
\begin{proof} For any points $x,x'\in \X$, we successively estimate
\begin{align*}
    \|\nabla_i f_i^\query(x)-\nabla_if_i^\query(x')\|& \leq 
   \E_{w_i\sim U_i}\left[\|\nabla_i f_i(x_i+ \delta w_i,x_{-i}+ \delta
   w_{-i})\right.\left.-\nabla_i f_i(x_i'+ \delta w_i,x_{-i}'+\delta w_{-i})\|\right]\leq \beta \|x-x'\|.
\end{align*}
Thus $\nabla_i f_i^\query$ is $\beta$-Lipschitz continuous. Next, we estimate
\begin{align*}
    \|\nabla_i f_i(x)-\nabla_i f^{\delta}_i(x)\|&\leq \E_{w\sim U_i}\left[
\|\nabla_i f_i(x+\delta w)-\nabla_i f_i(x)\|\right]\\
&\leq\beta \E_{w\sim U_i}
\sqrt{\|w_i\|^2+(n-1)}\\
&\leq \beta\delta\sqrt{\num}.
\end{align*}
Therefore we deduce 
\[\|g(x)-g^{\delta}(x)\|=\sqrt{\sum_{i=1}^\num \|\nabla_i f_i(x)-\nabla_i
f^{\delta}_i(x)\|^2}\leq \beta\delta \num,\]
as claimed.
\end{proof}

Observe that in the single player case $\num=1$, the function $f^{\delta}$ is trivially $\alpha$-strongly convex for any $\delta>0$. 
For general $\num>1$, the following lemma shows that the perturbed map $g^{\delta}$ is strongly monotone for all sufficiently small $\delta$.

\begin{lemma}[Strong monotonicity of the smoothed game]\label{lem:strong_mon} 
Choose  $\delta\leq \frac{c\alpha}{L \num^{3/2}}$  for some constant $c\in (0,1)$. Then the gradient map $g^{\delta}$ is strongly monotone over $\X$ with parameter $(1-c)\alpha$. 
\end{lemma}
\begin{proof}
Fix an index $i$ and let us first estimate the Lipschitz constant of the difference map 
\[H_i(x):= \nabla_if_i^\query(x)-\nabla_if_i(x).\]
To this end, we compute 
$$\nabla H_i(x)=\E_{w\sim U_i}  [\nabla (\nabla_i f_i)(x+\delta w)-
\nabla(\nabla_i f_i)(x)].$$
Taking into account that the map $x\mapsto \nabla(\nabla_i f_i)(x)$ is $L$-Lipschitz continuous, we deduce 
\begin{align*}
    \|\nabla H_i(x)\|_{\rm op}& \leq \E_{w\sim U_i}\left[\|(\nabla (\nabla_i
f_i)(x+\delta w)- \nabla(\nabla_i f_i)(x))\|_{\rm op}\right]\leq \delta L
\E_{w\sim U_i}\|w\|\leq \delta L\sqrt{\num}.
\end{align*}
Thus the map $H_i$ is Lipschitz continuous with parameter $\delta L\sqrt{\num}$. We therefore compute
\begin{align*}
    \la {g}^\query({x})-{g}^\query({x'}),x-{x'}\ra&= \sum_{i=1}^\num \langle \nabla_i  f_i^{\delta}(x)-\nabla_i f_i^{\delta}(x'),x_i-x_i'\rangle\\
    &=\sum_{i=1}^\num \langle \nabla_i f_i(x)-\nabla_i
    f_i(x'),x_i-x_i'\rangle-\sum_{i=1}^\num\langle H_i(x')-H_i(x),x_i-x'_i \rangle\\
    &\geq (\alpha-L \num^{3/2}\query)                \|x-x'\|^2.
\end{align*}
The proof is complete.
\end{proof}

The last ingredient, summarized in Lemma~\ref{lem:key_perturb}, is to estimate the distance between the Nash equilibria of the original and the perturbed games. Henceforth, let $x_\delta^{\star}$ be the Nash equilibrium of the game with losses $f_i^{\delta}$ over the set $(1-\delta)\X$. 

\begin{lemma}[Distance between equilibria]\label{lem:key_perturb}
Choose any  $\delta< \min\left\{r,\frac{\alpha}{L \num^{3/2}}\right\}$. Then the following estimate holds, 
\[\|x^\star-x_{\delta}^{\star}\|\leq
\delta\left(\left(1+\frac{\beta\sqrt{\num}}{\alpha}\right)\|x^\star\|+\frac{{\beta}\num}{\alpha}\right).\]
\end{lemma}
\begin{proof}
Lemma~\ref{lem:strong_mon}  and our choice of $\delta$ ensures that $g^{\delta}$ is strongly monotone and therefore that $x_\delta^{\star}$ is well-defined. There are two sources of perturbation: one replacing $\X$ with $(1-\delta)\X$ and the other in replacing $f_i$ with $f_i^{\delta}$. We deal with these in turn. To this end, set $\gamma:=1-\delta$ and let $\tilde x$ be the Nash equilibrium of the original game defined by the losses $f_i$ over the shrunken set $\gamma\X$. Thus
$\tilde x$ satisfies the inclusion \[0\in g(\tilde x)+N_{\gamma\X}(\tilde x),\]
where $N_{\gamma\X}(\tilde x)$ denotes the normal cone to $\gamma\X$ at $\tilde x$. The triangle inequality directly gives
\begin{equation}\label{eqn:main_bd}
\|x^\star-x_{\delta}^{\star}\|\leq \|x^\star-\tilde x\|+\|\tilde x-x_{\delta}^{\star}\|.
\end{equation}
Let us bound the first term on the right side of \eqref{eqn:main_bd}. To this
end, since the map $x\mapsto g(x)+N_{\gamma \X}(x)$ is $\alpha$-strongly monotone, we deduce
\begin{equation}\label{eqn:sr_mon}
\alpha\|\tilde x-\gamma x^\star\|\leq \textrm{dist}(0,g(\gamma x^\star)+N_{\gamma\X}(\gamma x^\star)).
\end{equation}
Let us estimate the right-hand side of \eqref{eqn:sr_mon}. Since $x^\star$ is a Nash
equilibrium of the original game over $\X$, the inclusion \[0\in
g(x^\star)+N_{\X}(x^\star)\] holds. Taking into account the identity $N_{\gamma\X}(\gamma x^\star)=N_{\X}( x^\star)$, we deduce
\begin{align*}
d(0,g(\gamma x^\star)+N_{\gamma\X}(\gamma x^\star))&= d(0,g(\gamma
x^\star)+N_{\X}(x^\star))\\
&\leq
\|g(\gamma x^\star)-g(x^\star)\|\\
&\leq \delta \beta\sqrt{\num}\|x^\star\|,
\end{align*}
where the last inequality follows from $g$ being $\beta\sqrt{\num}$-Lipschitz continuous.
Appealing to \eqref{eqn:sr_mon} and using the triangle inequality, we therefore deduce 
\begin{equation}    \label{eqn:tildestar}
\|x^\star-\tilde x\|\leq \|\tilde x-\gamma x^\star\|+\delta\|x^\star\|\leq
\delta(1+\tfrac{\beta\sqrt{\num}}{\alpha})\|x^\star\|.
\end{equation}
It remains to upper bound $\|\tilde x-x_{\delta}^{\star}\|$. The definition of $\tilde x$ as a Nash equilibrium ensures
\begin{equation}
\label{eq:original_nash_ineq}
\la -g(\tilde x),x-\tilde x\ra\leq 0,\ \ \forall \ x\in \gamma\mc{X}.
\end{equation}
Analogously, the definition of $x_\query^\ast$ as a Nash equilibrium ensures
\begin{equation}
\label{eq:smoothed_nash_ineq}
\la -g^\query(x^\star_\query),x-x^\star_\query\ra\leq 0,\ \ \forall \ x\in \gamma\mc{X}.  
\end{equation}
Then by strong monotonicity of the game and estimates~\eqref{eq:original_nash_ineq} and \eqref{eq:smoothed_nash_ineq}, we get
\begin{align*}
    \alpha\|\tilde x-x^\star_{\delta}\|^2&\leq \la g(\tilde x)-g(x^\star_{\delta}),\tilde x-x^\star_{\delta}\ra\\
    &\leq \la g^\query(x^\star_{\delta})-g(x^\star_{\delta}),\tilde x-x^\star_\query\ra\\
    &\leq \|g^\query(x^\star_\delta)-g(x^\star_{\delta})\|\cdot \|\tilde x-x^\star_{\delta}\|\\
    &\leq \beta\query \num\|\tilde x-x^\star_{\delta}\|,
\end{align*}
where the last inequality follows from Lemma~\ref{lem:smooth_game}.
Rearranging, we conclude \[\|\tilde x-x^\star_{\delta}\|\leq
\frac{{\beta\query}\num}{\alpha},\] which combined with \eqref{eqn:main_bd} and \eqref{eqn:tildestar} completes the proof.
\end{proof}

With the observations that the game defined by smoothed functions $f_i^\query$ is strongly monotone and the individual gradients of the smoothed loss functions are Lipschitz, we arrive at the following efficiency guarantee.

\begin{theorem}\label{thm:main}
	Suppose $\delta\leq \min\{r,\frac{\alpha}{2L \num^{3/2}}\}$ and set $\eta_t=\frac{2}{\alpha t}$. Then the following holds,
    \begin{align*}
    \E[ \|x^{t}-x^\star\|^2]&\leq  \frac{\max\{\delta^2\alpha^2
        \|x^{1}-x^\star_{\delta}\|^2,8F_* d^2
    \num\}}{\delta^2\alpha^2t}+2\delta^2\left(\left( 1+\frac{\beta\sqrt{\num}}{\alpha} \right)\|x^\star\|+\frac{{\beta}\num}{\alpha}\right)^2.
    \end{align*}
\end{theorem}

\begin{proof}
Using the inequality $(a+b)^2\leq 2a^2+2 b^2$ and Lemma~\ref{lem:key_perturb}, we estimate the one step progress
\begin{align*}
\frac{1}{2}\|x^{t+1}-x^\star\|^2&\leq \|x^{t+1}-x^\star_{\delta}\|^2+\|x^\star_\delta-x^\star\|^2\\
&\leq \|x^{t+1}-x^\star_{\delta}\|^2+
\delta^2\left(\left(1+\frac{\beta\sqrt{\num}}{\alpha}\right)\|x^\star\|+\frac{{\beta}\num^{3/2}}{\alpha}\right)^2.
\end{align*}
Next we continue with the standard estimate using non-expansiveness of the projection:
\begin{equation}\label{eqn:standard}
\begin{aligned}
 \E[\|x^{t+1}-x^\star_{\delta}\|^2]&=\E\left[\Big\|\proj_{(1-\delta)\X}(x^t-\eta_t\widehat{g}_{t}(x^t))-x^\star_{\delta}\Big\|^2\right]\\
 &\leq \E[\|x^t-x^\star_{\delta}-\eta_t\widehat{g}_{t}(x^t)\|^2]\\
 &\leq \E[\|x^t-x_{\delta}^*\|^2]+\eta^2_t\E[\|\widehat{g}_{t}(x^t)\|^2]-2\eta_t\E[\la  \widehat{g}_{t}(x^t),x^t-x^\star_{\delta}\ra] \\
  &\leq \E[\|x^t-x_{\delta}^*\|^2]-2\eta_t[\la
  g^{\delta}(x^t),x^t-x^\star_{\delta}\ra]+\eta^2_t\E[\|\widehat{g}_{t}(x^t)\|^2].
\end{aligned}
\end{equation}
Since $g^{\delta}$ is strongly monotone with parameter $\frac{\alpha}{2}$ (Lemma~\ref{lem:strong_mon}) and $x_{\delta}^*$ is the Nash equilibrium of the smoothed game, we deduce
\[\la  g^{\delta}(x^t),x^t-x^\star_{\delta}\ra\geq \la
    g^{\delta}(x^t)-g^{\delta}(x_{\delta}^*),x^t-x^\star_{\delta}\ra\geq
\frac{\alpha}{2}\|x^t-x^\star_{\delta}\|^2.\]
Returning to \eqref{eqn:standard}, we conclude
\[\E[ \|x^{t+1}-x^\star_{\delta}\|^2]\leq (1-\alpha\eta_t)\E\|x^t-x_{\delta}^*\|^2+\eta^2_t\frac{F_*d^2 \num}{\delta^2}
.\] A standard inductive argument shows \[\E[ \|x^{t}-x^\star_{\delta}\|^2]\leq
\frac{\max\{\delta^2\alpha^2 \|x^{1}-x^\star_{\delta}\|^2,8F_* d^2
\num\}}{\delta^2\alpha^2t}\] for all $t\geq 1$. The proof is complete.
\end{proof}

\paragraph{\textbf{Main Result.}}
The following is now the formal statement and proof of Theorem~\ref{thm:informal_main}.
\begin{corollary}[Main Result]\label{cor:main}
Fix a target accuracy \[\varepsilon< ((\alpha+\beta\sqrt{\num})R+\beta
\num)^2\cdot\min\{\tfrac{1}{L^2 \num^{3}},\tfrac{4r^2}{\alpha^2}\},\] 
 and set
 \[\delta=\frac{\alpha\sqrt{\varepsilon/4}}{(\alpha+\beta\sqrt{\num})R+{\beta}\num}\]
 and $\eta_t=\frac{2}{\alpha t}$. Then, the estimate 
 $\mathbb{E}[\|x^t-x^\star\|^2]\leq \varepsilon$ holds for all
 \[t\geq  \frac{\max\{32\alpha^4\varepsilon
     R^2,64((\alpha+\beta\sqrt{\num})R+\beta \num)^2F_* d^2
 \num\}}{\alpha^4\varepsilon^2}.\]
In the single player setting $\num=1$, the conclusion holds for any
$\varepsilon<4r^2((1+\frac{\beta}{\alpha}))R+\frac{\beta}{\alpha} )^2$ and
Assumption~\ref{assump:convex} (item c) may be dropped.
\end{corollary}
\begin{proof}
The assumed upper bound on $\varepsilon$  directly implies $\delta\leq \frac{\alpha}{2L \num^{3/2}}$ and $\delta<r$. An application of Theorem~\ref{thm:main} yields the estimate 
\[\E[ \|x^{t}-x^\star\|^2]\leq  \frac{\max\{\delta^2\alpha^2
    \|x^1-x^{*}_{\delta}\|^2,8F_* d^2
\num\}}{\delta^2\alpha^2t}+\frac{\varepsilon}{2}.\]
Setting the right side to $\varepsilon$, solving for $t$, and using the trivial
upper bound $\|x^1-x^{*}_{\delta}\|\leq 2R$ completes the proof. The claims for
the single player setting $\num=1$ follows by noting that the conclusions of
Lemmas~\ref{lem:strong_mon}--\ref{lem:key_perturb} and Theorem~\ref{thm:main} hold for any $\delta\in (0,r)$, since $f^{\delta}$ is strongly convex for any $\delta>0$.
\end{proof}

Observe that the sequence of actions that each player $i$ is actually
playing is $x_i^t+\delta v_i^t$. The efficiency guarantee in the preceding
corollary remains intact for this sequence with a slightly different choice of parameters. Indeed, 
the following estimate holds: 
    \begin{align*}
\frac{1}{2}\|x^{t+1}+\delta v^{t+1}-x^\star\|^2&\leq
    \|x^{t+1}-x_\delta^\star\|+\|x^\star_\delta-x^\star\|^2+\delta^2\|v^{t+1}\|^2
    \\
    &\leq  \|x^{t+1}-x^\star_{\delta}\|^2+
    \delta^2\left(\left(1+\frac{\beta\sqrt{\num}}{\alpha}\right)\|x^\star\|+\frac{{\beta}\num^{3/2}}{\alpha}+\sqrt{n}\right)^2.
    \end{align*}
Hence, if 
$\delta=
\frac{\alpha\sqrt{\varepsilon/4}}{\left(\left( \alpha+{\beta\sqrt{\num}}
\right)R+{\beta}\num+\alpha\sqrt{n}\right)}$
and 
$\varepsilon<((\alpha+\beta\sqrt{\num})R+\beta
\num+\alpha \sqrt{n})^2\cdot\min\{\tfrac{1}{L^2
\num^{3}},\tfrac{4r^2}{\alpha^2}\}$,
then an analogous argument to what is presented in the proof of
Corollary~\ref{cor:main} holds. In fact, one can choose parameters such that
both guarantees holds as is summarized in the following corollary. 
\begin{corollary}
    Fix a target accuracy 
    $\textstyle\varepsilon<((\alpha+\beta\sqrt{\num})R+\beta
\num)^2\cdot
\min\left\{\frac{1}{L^2n^3},\frac{4r^2}{\alpha^2}\right\}$,
 query radius
 $\delta= \frac{\alpha\sqrt{\varepsilon/4}}{(\alpha+\beta\sqrt{n})R+\beta n)}$,
and step-size $\eta_t=\frac{2}{\alpha t}$. Then, the estimates
$\mb{E}[\|x^t-x^\star\|^2]\leq \varepsilon$ and $\mb{E}[\|x^t+\delta v^t-x^\star\|^2]\leq \varepsilon$ hold for all 
\[t\geq\frac{\max\{32\alpha^4\varepsilon
     R^2,64\cdot ((\alpha+\beta\sqrt{\num})R+\beta
\num+\sqrt{n}\alpha)^2\cdot F_* d^2
 \num\}}{\alpha^4\varepsilon^2}. \]
\end{corollary}

In practice it may be advantageous to not specify $\varepsilon$ at the onset and instead allow the algorithm to run indefinitely. This can be easily achieved without sacrificing efficiency simply by restarting the algorithm periodically while shrinking $\delta$ by a constant fraction. The resulting process and its convergence rate is summarized in the following corollary.

\begin{corollary}[Efficiency without target accuracy]
Define the following constants: 
\[A=\frac{8F_* d^2 \num}{\alpha^2},\ \ 
    B=2\left(\left(1+\frac{\beta\sqrt{\num}}{\alpha}\right)R+\frac{{\beta}\num}{\alpha}\right)^2,~~\]
    \[\delta_1=\min\left\{r,\frac{\alpha}{2L
        \num^{3/2}}\right\},\ \ 
  T_1=\left\lceil\max\left\{\frac{A}{B\delta_{1}^4},
    \frac{4R^2}{B\delta_1^2}\right\}\right\rceil.\]
Fix a fraction $q\in (0,1)$, and set $y^{0}=x^0$ 
and $\eta_t=\frac{2}{\alpha t}$ for each index $t\geq 1$. 
Consider the following process: 
\begin{equation}
\left\{
\begin{aligned}
y^{k}&={\normalfont{\texttt{DFO}}}(y^{k-1}, \{\eta_t\}_{t\geq 1},\delta_k, T_k)\\
\delta_{k+1}&= q\cdot \delta_k\\
T_{k+1}&=\left\lceil\max\left\{\frac{A}{B\delta_{k+1}^4}, \frac{4R^2}{B\delta_{k+1}^2}\right\}\right\rceil
\end{aligned}\right\}.
\end{equation}
For every $k\geq 1$, the iterate $y^{k}$ satisfies
\[\E[\|y^{k}-x^\star\|^2]\leq \varepsilon_k,\] where
$\varepsilon_k:=2B\delta_1^2 q^{2(k-1)}$, while the total number of steps of
Algorithm~\ref{algo:bravo} needed to generate $y^{k}$ is at most
\[1+\frac{1}{2}\log\left(\frac{2B\delta_1^2}{\varepsilon_k}\right)+\frac{4 B A
    q^{-4}\delta_1^2}{q^{-4}-1}\cdot
    \varepsilon^{-2}_k+\frac{8R^2q^{-2}\delta_1^{-2}}{q^{-2}-1}\varepsilon_{k}^{-1}
.\] In the single player setting (i.e., $\num=1$), we may instead set $\delta_1=r$.
\end{corollary}
\begin{proof}
Theorem~\ref{thm:main} directly guarantees 
	\[\E[ \|y^{k}-x^\star\|^2]\leq \frac{\max\{A, 4R^2\delta_k^2\}}{\delta_k^2
    T_k}+B\delta_k^2\leq 2 B\delta_k^2.\]
    The total number of steps of Algorithm~\ref{algo:bravo} needed to generate
    $y_k$ is bounded as follows:
\begin{align*}
\sum_{i=1}^k T_i&\leq\sum_{i=1}^k \left\lceil\frac{A}{B}\delta_1^{-4}\cdot q^{-4(i-1)}+\frac{4R^2}{B}\delta_1^{-2}\cdot q^{-2(i-1)}\right\rceil
\\
&\leq k+\frac{A\delta_1^{-4}}{B}\sum_{j=0}^{k-1} q^{-4j}+k+\frac{4R^2\delta_1^{-2}}{B} \sum_{j=0}^{k-1}q^{-2j}\\
&\leq k+\frac{A\delta_1^{-4}}{B(q^{-4}-1)}\cdot q^{-4k}+\frac{4R^2\delta_1^{-2}}{B(q^{-2}-1)} q^{-2k}.
\end{align*}	
Rewriting the right-side in terms of $\varepsilon_k$ completes the proof.
\end{proof}

 \section{Discussion and Future Directions}

A promising future research direction is to examine the benefits of using a
two-point (and multi-point) estimate for the gradient.  In a two-point method,
once a random perturbation direction is chosen, two function evaluations are
performed along that direction, for example $f(x+\delta u)-f(x-\delta u)$, as
examined by Agarwal et al.~\cite{agarwal2010optimal} and by Nesterov and Spokoiny~\cite{nesterov2017random}. 
The use of this symmetric two-point estimate can improve the convergence rate
via improving the constants---however, the dependence of the rate on problem
dimension $d$ still is not optimal. 
We remark that this symmetric expression yields an unbiased estimate of the gradient of $f$, and its extension to the game setting is straightforward; i.e., the proof approach in the current paper can be applied.

One point to note when considering the extension of two-point methods to the
 multiplayer (game) setting is to ensure each player can indeed evaluate their
 cost function $f_i(x_i,x_{-i})$ at two points---player $i$ can certainly vary
 its own action $x_i$, but can it do so while keeping the other players' actions
 $x_{-i}$ fixed? This
type of method requires more explicit coordination between players; however, there are 
practical 
settings in machine learning where such coordination is possible. One example is
\emph{multiplayer performative prediction} which is introduced in the recent
work~\cite{narang2022multiplayer}, wherein players have a loss function oracle
and observe their competitors' actions. This would enable players to form
estimates using query responses of the form $f_i(x_i+\delta
u_i,x_{-i})$. For instance, the two point method of \citet{agarwal2010optimal}
takes the form
\[\hat{g}_i=\frac{d_i}{2\delta}\left(f_i(x_i+\delta
u_i,x_{-i})-f_i(x_i-\delta
u_i,x_{-i})\right).\]

The analysis in \cite{duchi2015optimal} improves the rate in
\cite{agarwal2010optimal} significantly by a factor of $\sqrt{d}$, by employing 
a one-sided two-point estimate of the gradient. This introduces additional bias in the gradient estimate 
that they cleverly handle.
It is an interesting direction to adapt this algorithm to the game setting. 
However, due to the asymmetry and hence bias in the estimate of \cite{duchi2015optimal}, a new proof approach is needed. 
We leave this topic to future work.

\bibliographystyle{IEEEtranN}
\bibliography{2022coltrefs}

\begin{thebibliography}{23}
\providecommand{\natexlab}[1]{#1}
\providecommand{\url}[1]{#1}
\csname url@samestyle\endcsname
\providecommand{\newblock}{\relax}
\providecommand{\bibinfo}[2]{#2}
\providecommand{\BIBentrySTDinterwordspacing}{\spaceskip=0pt\relax}
\providecommand{\BIBentryALTinterwordstretchfactor}{4}
\providecommand{\BIBentryALTinterwordspacing}{\spaceskip=\fontdimen2\font plus
\BIBentryALTinterwordstretchfactor\fontdimen3\font minus
  \fontdimen4\font\relax}
\providecommand{\BIBforeignlanguage}[2]{{%
\expandafter\ifx\csname l@#1\endcsname\relax
\typeout{** WARNING: IEEEtranN.bst: No hyphenation pattern has been}%
\typeout{** loaded for the language `#1'. Using the pattern for}%
\typeout{** the default language instead.}%
\else
\language=\csname l@#1\endcsname
\fi
#2}}
\providecommand{\BIBdecl}{\relax}
\BIBdecl

\bibitem[Bravo et~al.(2018)Bravo, Leslie, and Mertikopoulos]{bravo2018bandit}
M.~Bravo, D.~S. Leslie, and P.~Mertikopoulos, ``Bandit learning in concave $ n
  $-person games,'' \emph{Proceedings of the Conference on Neural Information
  Processing Systems (NeurIPS)}, 2018.

\bibitem[Madry et~al.(2018)Madry, Makelov, Schmidt, Tsipras, and
  Vladu]{madry2018towards}
A.~Madry, A.~Makelov, L.~Schmidt, D.~Tsipras, and A.~Vladu, ``Towards deep
  learning models resistant to adversarial attacks,'' \emph{Proceedings of the
  International Conference Learning Representation (ICLR)}, 2018.

\bibitem[Fiez et~al.(2020)Fiez, Chasnov, and Ratliff]{fiez2020implicit}
T.~Fiez, B.~Chasnov, and L.~Ratliff, ``Implicit learning dynamics in
  stackelberg games: Equilibria characterization, convergence analysis, and
  empirical study,'' in \emph{International Conference on Machine
  Learning}.\hskip 1em plus 0.5em minus 0.4em\relax PMLR, 2020, pp. 3133--3144.

\bibitem[Goodfellow et~al.(2014)Goodfellow, Pouget-Abadie, Mirza, Xu,
  Warde-Farley, Ozair, Courville, and Bengio]{goodfellow2014generative}
I.~Goodfellow, J.~Pouget-Abadie, M.~Mirza, B.~Xu, D.~Warde-Farley, S.~Ozair,
  A.~Courville, and Y.~Bengio, ``Generative adversarial nets,'' \emph{Advances
  in neural information processing systems}, vol.~27, 2014.

\bibitem[Narang et~al.(2022)Narang, Faulkner, Drusvyatskiy, Fazel, and
  Ratliff]{narang2022multiplayer}
A.~Narang, E.~Faulkner, D.~Drusvyatskiy, M.~Fazel, and L.~J. Ratliff,
  ``Learning in strongly monotone decision-dependent games,'' \emph{Proceedings
  of the Artificial Intelligence and Statistics Conference (AIStats)}, 2022.

\bibitem[Ratliff et~al.(2016)Ratliff, Burden, and
  Sastry]{ratliff2016characterization}
L.~J. Ratliff, S.~A. Burden, and S.~S. Sastry, ``On the characterization of
  local nash equilibria in continuous games,'' \emph{IEEE transactions on
  automatic control}, vol.~61, no.~8, pp. 2301--2307, 2016.

\bibitem[Zhou et~al.(2017)Zhou, Mertikopoulos, Moustakas, Bambos, and
  Glynn]{zhou2017mirror}
Z.~Zhou, P.~Mertikopoulos, A.~L. Moustakas, N.~Bambos, and P.~Glynn, ``Mirror
  descent learning in continuous games,'' in \emph{Proceedings of the 56th IEEE
  Annual Conference on Decision and Control (CDC)}, 2017, pp. 5776--5783.

\bibitem[Ratliff and Fiez(2020)]{ratliff2020adaptive}
L.~J. Ratliff and T.~Fiez, ``Adaptive incentive design,'' \emph{IEEE
  Transactions on Automatic Control}, vol.~66, no.~8, pp. 3871--3878, 2020.

\bibitem[Li and Marden(2011)]{li2011designing}
N.~Li and J.~R. Marden, ``Designing games for distributed optimization,'' in
  \emph{Proceedings of the 50th IEEE Conference on Decision and Control}, 2011.

\bibitem[Yekkehkhany et~al.(2021)Yekkehkhany, Feng, and
  Lavaei]{yekkehkhanyadversarial}
A.~Yekkehkhany, H.~Feng, and J.~Lavaei, ``Adversarial attacks on computation of
  the modified policy iteration method,'' in \emph{Proceedings of the IEEE
  Conference on Decision and Control}, 2021.

\bibitem[Zhang et~al.(2020)Zhang, Yang, and Ba\c{s}ar]{Zhang2019MultiAgentRL}
K.~Zhang, Z.~Yang, and T.~Ba\c{s}ar, ``Multi-agent reinforcement learning: A
  selective overview of theories and algorithms,'' \emph{Springer Studies in
  Systems, Decision and Control, Handbook on RL and Control}, 2020.

\bibitem[Zhang et~al.(2019)Zhang, Yang, and Basar]{zhang2019policy}
K.~Zhang, Z.~Yang, and T.~Basar, ``Policy optimization provably converges to
  nash equilibria in zero-sum linear quadratic games,'' \emph{Advances in
  Neural Information Processing Systems}, vol.~32, 2019.

\bibitem[Savas et~al.(2019)Savas, Gupta, Ornik, Ratliff, and
  Topcu]{savas2019incentive}
Y.~Savas, V.~Gupta, M.~Ornik, L.~J. Ratliff, and U.~Topcu, ``Incentive design
  for temporal logic objectives,'' in \emph{Proceedings of the 58th IEEE
  Conference on Decision and Control}, 2019, pp. 2251--2258.

\bibitem[Fudenberg et~al.(1998)Fudenberg, Drew, Levine, and
  Levine]{fudenberg1998theory}
D.~Fudenberg, F.~Drew, D.~K. Levine, and D.~K. Levine, \emph{The theory of
  learning in games}.\hskip 1em plus 0.5em minus 0.4em\relax MIT press, 1998,
  vol.~2.

\bibitem[Cesa-Bianchi and Lugosi(2006)]{cesa2006prediction}
N.~Cesa-Bianchi and G.~Lugosi, \emph{Prediction, learning, and games}.\hskip
  1em plus 0.5em minus 0.4em\relax Cambridge university press, 2006.

\bibitem[Agarwal et~al.(2010)Agarwal, Dekel, and Xiao]{agarwal2010optimal}
A.~Agarwal, O.~Dekel, and L.~Xiao, ``Optimal algorithms for online convex
  optimization with multi-point bandit feedback.'' in \emph{Proceedings of the
  Conference on Learning Theory (COLT)}, 2010, pp. 28--40.

\bibitem[Shamir(2017)]{shamir2017optimal}
O.~Shamir, ``An optimal algorithm for bandit and zero-order convex optimization
  with two-point feedback,'' \emph{The Journal of Machine Learning Research},
  vol.~18, no.~1, pp. 1703--1713, 2017.

\bibitem[Flaxman et~al.(2005)Flaxman, Kalai, and McMahan]{flaxman2004online}
A.~D. Flaxman, A.~T. Kalai, and H.~B. McMahan, ``Online convex optimization in
  the bandit setting: gradient descent without a gradient,'' \emph{Proceedings
  of the sixteenth annual ACM-SIAM Symposium on Discrete Algorithms (SODA)},
  2005.

\bibitem[Nesterov and Spokoiny(2017)]{nesterov2017random}
Y.~Nesterov and V.~Spokoiny, ``Random gradient-free minimization of convex
  functions,'' \emph{Foundations of Computational Mathematics}, vol.~17, no.~2,
  pp. 527--566, 2017.

\bibitem[Tatarenko and Kamgarpour(2020)]{tatarenko2020bandit}
T.~Tatarenko and M.~Kamgarpour, ``Bandit online learning of nash equilibria in
  monotone games,'' \emph{arXiv preprint arXiv:2009.04258}, 2020.

\bibitem[Tatarenko and Kamgarpour(2019)]{tatarenko2019learning}
------, ``{Learning Nash equilibria in monotone games},'' in \emph{Proceedings
  of the 58th IEEE Conference on Decision and Control (CDC)}, 2019, pp.
  3104--3109.

\bibitem[Rosen(1965)]{rosen1965existence}
J.~B. Rosen, ``Existence and uniqueness of equilibrium points for concave
  n-person games,'' \emph{Econometrica: Journal of the Econometric Society},
  pp. 520--534, 1965.

\bibitem[Duchi et~al.(2015)Duchi, Jordan, Wainwright, and
  Wibisono]{duchi2015optimal}
J.~C. Duchi, M.~I. Jordan, M.~J. Wainwright, and A.~Wibisono, ``Optimal rates
  for zero-order convex optimization: The power of two function evaluations,''
  \emph{IEEE Transactions on Information Theory}, vol.~61, no.~5, pp.
  2788--2806, 2015.

\end{thebibliography}

\end{document}